\documentclass[sigconf]{acmart}
\AtBeginDocument{%
  }

\setcopyright{acmlicensed}
\copyrightyear{2018}
\acmYear{2018}
\acmDOI{XXXXXXX.XXXXXXX}
\acmConference[Conference acronym 'XX]{Make sure to enter the correct
  conference title from your rights confirmation email}{June 03--05,
  2018}{Woodstock, NY}
\acmISBN{978-1-4503-XXXX-X/2018/06}
\usepackage{multirow}
\usepackage{graphicx}
\usepackage{subcaption}

\begin{document}

\title{Differentiable Geometric Indexing for End-to-End Generative Retrieval}

\author{Xujing Wang}
\affiliation{%
  \institution{Xidian University}
  \city{Xi'an}
  \state{Shaanxi}
  \country{China}
}
\email{xjwong@stu.xidian.edu.cn}

\author{Yufeng Chen}
\authornote{Project leader}
\affiliation{%
  \institution{Alibaba}
  \city{Hangzhou}
  \state{Zhejiang}
  \country{China}
}
\email{chenyufeng.cyf@alibaba-inc.com}

\author{Boxuan Zhang}
\affiliation{%
  \institution{Alibaba}
  \city{Hangzhou}
  \state{Zhejiang}
  \country{China}
}
\email{boxuan.zbx@alibaba-inc.com}

\author{Jie Zhao}
\affiliation{%
  \institution{Xidian University}
  \city{Xi'an}
  \state{Shaanxi}
  \country{China}
}

\author{Chao Wei}
\affiliation{%
  \institution{Alibaba}
  \city{Hangzhou}
  \state{Zhejiang}
  \country{China}
}

\author{Cai Xu}
\affiliation{%
  \institution{Alibaba}
  \city{Hangzhou}
  \state{Zhejiang}
  \country{China}
}
\email{cxu@xidian.edu.cn}

\author{Ziyu Guan}
\authornote{Corresponding author}
\affiliation{%
  \institution{Xidian University}
  \city{Xi'an}
  \state{Shaanxi}
  \country{China}
}
\email{zyguan@xidian.edu.cn}

\author{Wei Zhao}
\affiliation{%
  \institution{Xidian University}
  \city{Xi'an}
  \state{Shaanxi}
  \country{China}
}

\author{Weiru Zhang}
\affiliation{%
  \institution{Alibaba}
  \city{Hangzhou}
  \state{Zhejiang}
  \country{China}
}

\author{Xiaoyi Zeng}
\affiliation{%
  \institution{Alibaba}
  \city{Hangzhou}
  \state{Zhejiang}
  \country{China}
}
\renewcommand{\shortauthors}{Wang et al.}

\begin{abstract}

Generative Retrieval (GR) has emerged as a promising paradigm to unify indexing and search within a single probabilistic framework. However, existing approaches suffer from two intrinsic conflicts: (1) an Optimization Blockage, where the non-differentiable nature of discrete indexing creates a gradient blockage, decoupling index construction from the downstream retrieval objective; and (2) a Geometric Conflict, where standard unnormalized inner-product objectives induce norm-inflation instability, causing popular ``hub'' items to geometrically overshadow relevant long-tail items.

To systematically resolve these misalignments, we propose \textbf{Differentiable Geometric Indexing (DGI)}. 
First, to bridge the optimization gap, DGI enforces Operational Unification. It employs Soft Teacher Forcing via Gumbel-Softmax to establish a fully differentiable pathway, combined with Symmetric Weight Sharing to effectively align the quantizer's indexing space with the retriever's decoding space.
Second, to restore geometric fidelity, DGI introduces Isotropic Geometric Optimization. We replace inner-product logits with scaled cosine similarity on the unit hypersphere to effectively decouple popularity bias from semantic relevance.

Extensive experiments on large-scale industry search datasets and online e-commerce platform demonstrate that DGI outperforms competitive sparse, dense, and generative baselines. Notably, DGI exhibits superior robustness in long-tail scenarios, validating the necessity of harmonizing structural differentiability with geometric isotropy.

\end{abstract}



\maketitle
\begin{figure}[hbpt]
  \centering
  \includegraphics[width=\linewidth]{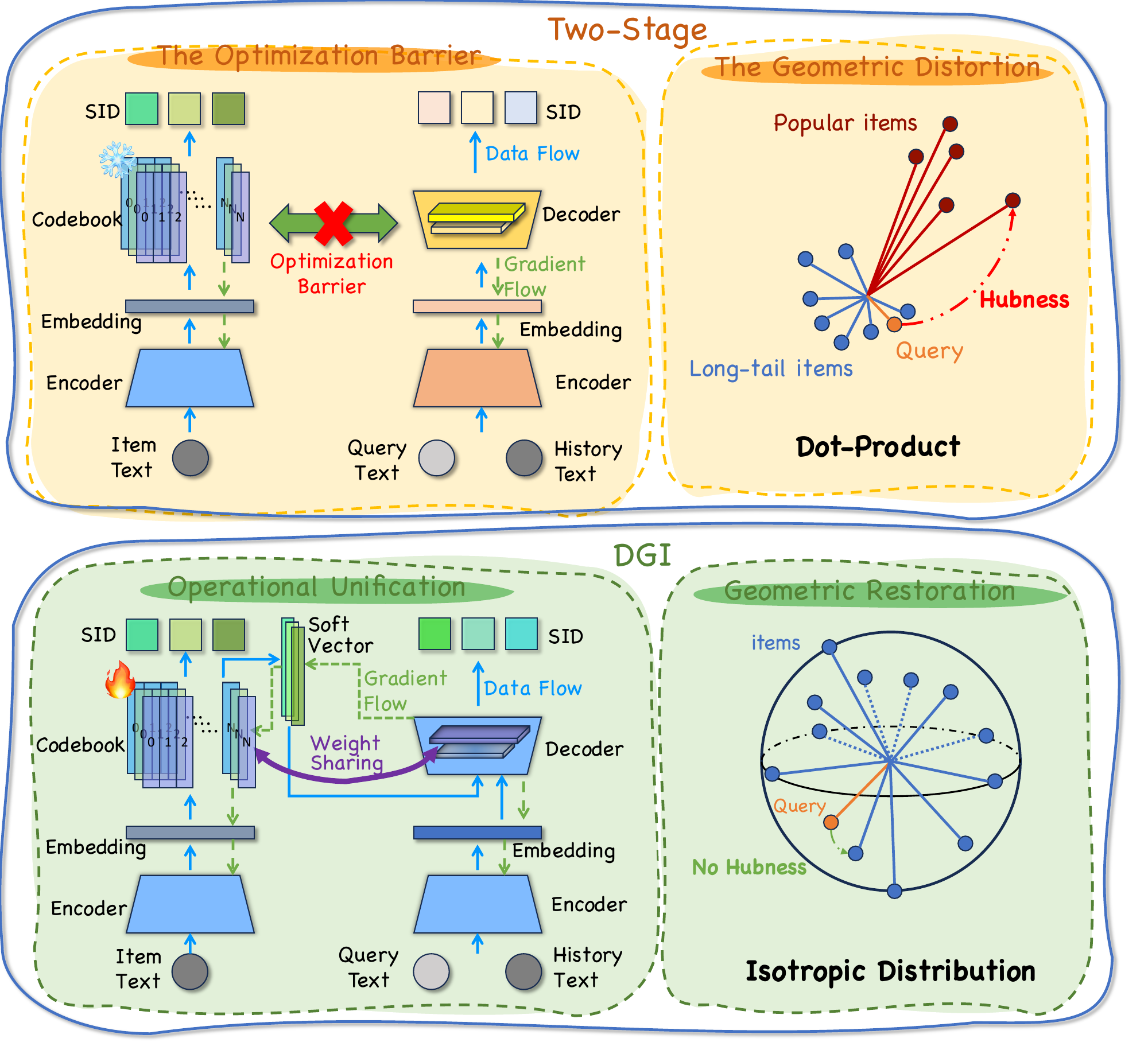}
  \caption{Illustration of the Structural Mismatch and Geometric Anisotropy in existing GR frameworks compared to our DGI.}
  \label{fig1}
\end{figure}

\section{Introduction}

Large-scale industrial search systems serve as the information backbone of modern platforms, bridging user intent with massive item inventories. Traditional search systems typically follow a multi-stage ``retrieve-then-rank'' architecture. In the initial retrieval stage, a sparse \cite{karpukhin2020dense,huang2013learning,ni2022sentence} or dense \cite{robertson2009probabilistic,nogueira2019doc2query} retriever identifies a coarse candidate set from a massive corpus, which is subsequently refined by a ranking model. However, this pipeline suffers from a fundamental structural fragmentation. Specifically, the retrieval process relies on an external index (e.g., inverted lists or ANN graphs) for efficient lookup. This typically follows a disjoint two-stage paradigm: the representation model is trained first, and the index is built subsequently using the frozen item embeddings. As a result, the indexing topology is isolated from the model's gradient flow, creating an optimization gap.
Consequently, it is impossible to jointly optimize the retrieval and ranking stages end-to-end, as this non-learnable index blocks the propagation of optimization signals\cite{khandagale2025interactrank, zheng2024full}.

Recent advances in Large Language Models (LLMs) have motivated a pathbreaking attempt toward Generative Retrieval (GR) \cite{rare2025,chen2025unisearch}, which casts retrieval as a unified sequence generation problem by predicting item IDs autoregressively.
Typically, GR falls into the conditional generation framework and directly optimizes the probability of generating a relevant item's identifier ($ID$) conditioned on the input query $q$ and an optional user’s interaction history $h$ (i.e., maximizing $P(ID|q,h)$).
A large body of work has been proposed to construct semantically rich IDs of candidate items.
Pioneering works, like DSI \cite{tay2022transformer} and NCI\cite{wang2022neural}, assign structured IDs to candidate items using manually designed rules (e.g., hierarchical clustering). 

Though simple, these approaches suffer from a static vocabulary of IDs, which means that adding new items requires retraining the entire retrieval model.
Moreover, manually defined heuristic rules for ID generation\cite{rare2025,jdgr2025} struggle to comprehensively capture the semantic information embedded in items.
Building on this, state-of-the-art methods like TIGER \cite{tiger2023} and OneSearch \cite{onesearch2025} have converged toward Learnable Semantic Identifier (SID). By employing residual quantization (e.g., RQ-VAE \cite{rqvae} and RQ-Kmeans\cite{luo2025qarm}), discrete codes (identifier strings) can be directly derived from item contents. 
SID allows for capturing deep inter-item relationships, offering the potential to generalize to unseen items based on semantic similarity.

Despite these advances, we argue that current GR remain constrained by two fundamental issues that limit their potential:

\textbf{Optimization Blockage}: Prevailing works ~\cite{tiger2023,onesearch2025,genius2025} typically follow a two-stage paradigm. The indexer (i.e., the item encoder and quantizer) is independently trained for semantic preservation of items. And then it is freezed during the retriever training, as the discrete SIDs creates a fundamental gradient blockage between the optimizations of retriever and indexer.
Consequently, the retrieval loss cannot backpropagate to the indexer, leaving the SID suboptimal to the central retrieval task and preventing the retrieval system from distinguishing confusable and hard negative items. 
Though some efforts achieve joint-training~\cite{chen2025unisearch,etegrec2024} relying on surrogate estimators (e.g., STE \cite{bengio2013estimating}), they yield biased gradients.
Specifically, the STE directly copies the gradient of the variable after quantization to the gradient of the variable before quantization. Since the non-differentiable argmax function is used during the quantization process, this is essentially a biased gradient estimation.

\textbf{Geometric Conflict:} Ideally, the retrieval should ensure that the retrievability of an item is driven by relevance, without being distorted by popularity-induced bias~\cite{he2009learning,radovanovic2010hubs}. Real-world industrial retrieval scenarios (e.g., e-commerce and web search), however, are inherently characterized by severe long-tail distributions.
And the standard cross-entropy optimization on long-tail distribution data with unnormalized logits typically leads to norm-dominated ranking. 
Popular items acquire inflated norms that disproportionately influence the similarity score (as shown in~\ref{hubness}). 
This creates a ``hubness'' effect where high-frequency items geometrically overshadow semantically relevant long-tail items, even when they better match the query intent.

To systematically resolve these issues, we propose \textbf{Differentiable Geometric Indexing (DGI)}, a holistic framework grounded in two design pillars corresponding to the issues identified above: (1) \textbf{Operational Unification:} To bridge the \textit{Optimization Blockage}, we enable fully differentiable training via Soft Teacher Forcing with Gumbel-Softmax and Symmetric Weight Sharing. This strategy strictly aligns the quantizer's codebook space with the decoder's retrieval space, ensuring the index structure is dynamically reshaped by the retrieval objective.
(2) \textbf{Isotropic Geometric Optimization:} To resolve the \textit{Geometric Conflict}, we introduce a Riemannian perspective by replacing the dot-product with a Scaled Cosine objective on the unit hypersphere. This constraint explicitly decouples magnitude effects (often correlated with popularity) from semantic relevance (angle). By preventing norm inflation, we ensure that popular items do not suppress the visibility of long-tail items purely due to geometric artifacts, allowing for a more balanced and effective retrieval.

In summary, our key contributions are threefold:

First, we systematically identify two fundamental bottlenecks in existing Generative Retrieval paradigms: the \textit{Optimization Blockage}, where non-differentiable indexing blocks gradient propagation, and the \textit{Geometric Conflict}, where unnormalized dot products induce norm-dominated Hubness. 

Second, we propose a novel framework named Differentiable Geometric Indexing (DGI). By synergizing Operational Unification (via Soft Gradient Flow and Weight Sharing) with Isotropic Geometric Optimization, DGI bridges the discrete optimization gap and enforces strict geometric alignment with query intent.
   
Third, we validate DGI through extensive offline experiments and show significant performance gains compared with SOTA baselines. Additionally, a 7-day online A/B test in a commercial search advertising system yields a +1.27\% CTR increase and a +1.11\% RPM lift ($p < 0.001$), demonstrating its effectiveness in production environments.

\section{Related Works}
We review the evolution of retrieval paradigms and the geometric properties of representation learning.

\subsection{Retrieval Paradigms: From Static to Generative}
\textbf{Traditional Retrieval.} 
For decades, retrieval has followed an ``Index-then-Retrieve'' pipeline. Sparse Retrieval methods (e.g., BM25 \cite{robertson2009probabilistic}, DocT5Query \cite{nogueira2019doc2query}) utilize inverted indices for lexical matching, while Dense Retrieval (e.g., DSSM\cite{huang2013learning}, DPR \cite{karpukhin2020dense} and Sentence-T5 \cite{ni2022sentence}) employs dual-encoders for semantic matching via ANN search. 
Despite their effectiveness, both of them suffer from a \textit{separation constraint}: the index structure is static and decoupled from the final retrieval objective, limiting end-to-end optimization \cite{tay2022transformer, zheng2024full}.

\textbf{Generative Retrieval and Semantic Identifiers.}
Generative Retrieval (GR) unifies this process by reformulating retrieval as an autoregressive sequence-to-sequence generation task\cite{penha2025semantic,wu2024generative,chen-etal-2023-understanding,zhuang2022bridging}. Early approaches like DSI \cite{tay2022transformer} and NCI \cite{wang2022neural} utilized Atomic Identifiers, but suffered from static vocabularies and limited generalization.
To address this, the field has converged toward Semantic Identifiers (SIDs), which derive structured codes from item content. Existing SID methods can be categorized by their optimization strategy:

 \textit{Two-Stage Optimization:} Methods like TIGER \cite{tiger2023} and OneSearch \cite{onesearch2025} adopt an ``Index-then-Freeze'' strategy. They pre-train a quantizer to minimize reconstruction error and then freeze it. While stable, this renders the index oblivious to the downstream retrieval task.

\textit{Joint Optimization:} Recent works attempt to align indexing and retrieval. UniSearch \cite{chen2025unisearch} and ETEGRec \cite{etegrec2024} employ Joint Training, typically using the Straight-Through Estimator (STE) to bypass the non-differentiable quantization step. Others explore model-enhanced differentiable indexing\cite{tang2023semantic,zhang2023model,yang-etal-2023-auto,liu2025towards} or integrate ranking objectives \cite{li2024learning,li2024distillation,zhang2026conversational,lu2025dogr}. 

\textbf{The Gap.} While Joint Optimization improves alignment, STE-based methods rely on identity-function approximations, which can lead to biased gradients and optimization instability \cite{bengio2013estimating}. In contrast, DGI utilizes Gumbel-Softmax\cite{jang2016categorical,maddison2016concrete} relaxation with soft conditioning to enable end-to-end gradient flow without relying on STE.

\subsection{Geometry in Representation Learning}
Beyond optimization pathways, the geometry of the embedding space fundamentally constrains retrieval performance. 
In high-dimensional spaces with inner-product logits, models are prone to the Hubness Problem \cite{radovanovic2010hubs}. Frequent items tend to acquire excessively large norms to minimize softmax loss, leading to popularity-induced norm inflation \cite{zhu2025addressing}. This phenomenon distorts the semantic space, masking relevant long-tail items.
While spherical embeddings and margin-based losses have been standard in face recognition (e.g., ArcFace \cite{deng2019arcface}, CosFace \cite{wang2018cosface}) and contrastive learning (e.g., Alignment and Uniformity \cite{wang2020understanding}), their application to the \textit{discrete quantization} step in Generative Retrieval remains underexplored, particularly regarding the alignment between discrete codebook learning and continuous retrieval logits.
We adopt a Riemannian perspective to enforce geometric isotropy, aiming to mitigate the magnitude-semantic entanglement inherent in standard GR models.
\section{Preliminaries}

In this section, we formally define the generative retrieval task and establish the theoretical foundations of optimization on Riemannian manifolds, which serve as the geometric basis for our DGI framework.

\subsection{Generative Retrieval Task}

Let $\mathcal{Q} = \{q_1, \dots, q_M\}$ be a set of user queries and $\mathcal{I} = \{i_1, \dots, i_N\}$ be the corpus of items.
In a personalized retrieval scenario, each request consists of a current query $q \in \mathcal{Q}$ and a user interaction history $H = (h_1, \dots, h_L)$, where $h_k \in \mathcal{I}$ denotes historically interacted items.
The goal of Generative Retrieval (GR) is to learn a parameterized probabilistic model $P_{\theta}(i | q, H)$ that directly maps the context $(q, H)$ to the target item identifier.

Unlike traditional indexing which maps items to arbitrary integers, GR represents each item $i$ as a sequence of discrete semantic tokens $\mathbf{c}^{(i)} = (c_1, c_2, \dots, c_m)$, where each token $c_j$ is drawn from a codebook $\mathcal{V}$.
The retrieval probability is factorized autoregressively:
\begin{equation}
    P_{\theta}(i | q, H) = P_{\theta}(\mathbf{c}^{(i)} | q, H) = \prod_{j=1}^m P_{\theta}(c_j | q, H, \mathbf{c}_{<j})
\end{equation}
Our objective is to \textbf{jointly optimize} the parameters of two coupled modules:
(1) The Tokenization Network, denoted as $\mathcal{T}_{\phi}: \mathbb{R}^d \to \mathcal{V}^m$, which discretizes item representations into semantic codes;
and (2) The Generative Retriever parameterized by $\theta$, which learns to predict these codes.
By unifying these two objectives, we ensure that the index structure evolves dynamically to maximize retrievability.

\subsection{Optimization on Riemannian Manifolds}
To address the geometric anisotropy prevalent in high-dimensional Euclidean spaces, we ground our framework in Riemannian geometry.
We model the embedding space not as a flat vector space $\mathbb{R}^d$, but as a differentiable manifold $\mathcal{M}$. Specifically, we focus on the Unit Hypersphere $\mathbb{S}^{d-1}$, defined as:
\begin{equation}
    \mathbb{S}^{d-1} := \{\mathbf{x} \in \mathbb{R}^d : \|\mathbf{x}\|_2 = 1\}
\end{equation}
Optimization on this manifold requires redefining gradients and update rules to ensure parameters remain on the surface.

\textbf{Tangent Space and Riemannian Gradient.} For any point $\mathbf{x} \in \mathbb{S}^{d-1}$, the tangent space $T_{\mathbf{x}}\mathbb{S}^{d-1}$ is the vector space approximating the manifold locally at $\mathbf{x}$:
\begin{equation}
    T_{\mathbf{x}}\mathbb{S}^{d-1} := \{\mathbf{z} \in \mathbb{R}^d : \mathbf{x}^\top \mathbf{z} = 0\}
\end{equation}
Let $f: \mathbb{S}^{d-1} \to \mathbb{R}$ be a differentiable objective function. The \textit{Riemannian gradient} $\operatorname{grad} f(\mathbf{x})$ is the direction of steepest ascent within the tangent space. It is obtained by projecting the standard Euclidean gradient $\nabla f(\mathbf{x})$ onto $T_{\mathbf{x}}\mathbb{S}^{d-1}$:
\begin{equation}
    \label{eq:riem_grad}
    \operatorname{grad} f(\mathbf{x}) := (\mathbf{I} - \mathbf{x}\mathbf{x}^\top)\nabla f(\mathbf{x})
\end{equation}
This projection removes the radial component of the gradient (which changes the norm), ensuring that the update direction is purely tangential (changing only the direction).

\textbf{Retraction.} Since a step along the tangent vector leaves the manifold, a mapping is required to project the updated point back onto $\mathbb{S}^{d-1}$. This operation is formally known as a \textit{Retraction} $R_{\mathbf{x}}: T_{\mathbf{x}}\mathbb{S}^{d-1} \to \mathbb{S}^{d-1}$. A computationally efficient retraction for the sphere is simple normalization:
\begin{equation}
    \label{eq:retraction}
    R_{\mathbf{x}}(\mathbf{v}) := \frac{\mathbf{x} + \mathbf{v}}{\|\mathbf{x} + \mathbf{v}\|_2}
\end{equation}
In Section \ref{Isotropic}, we will demonstrate that our proposed Scaled Cosine objective implicitly implements this Riemannian optimization framework, thereby strictly enforcing geometric isotropy.

\begin{figure*}[t]
  \centering
  \includegraphics[width=\linewidth]{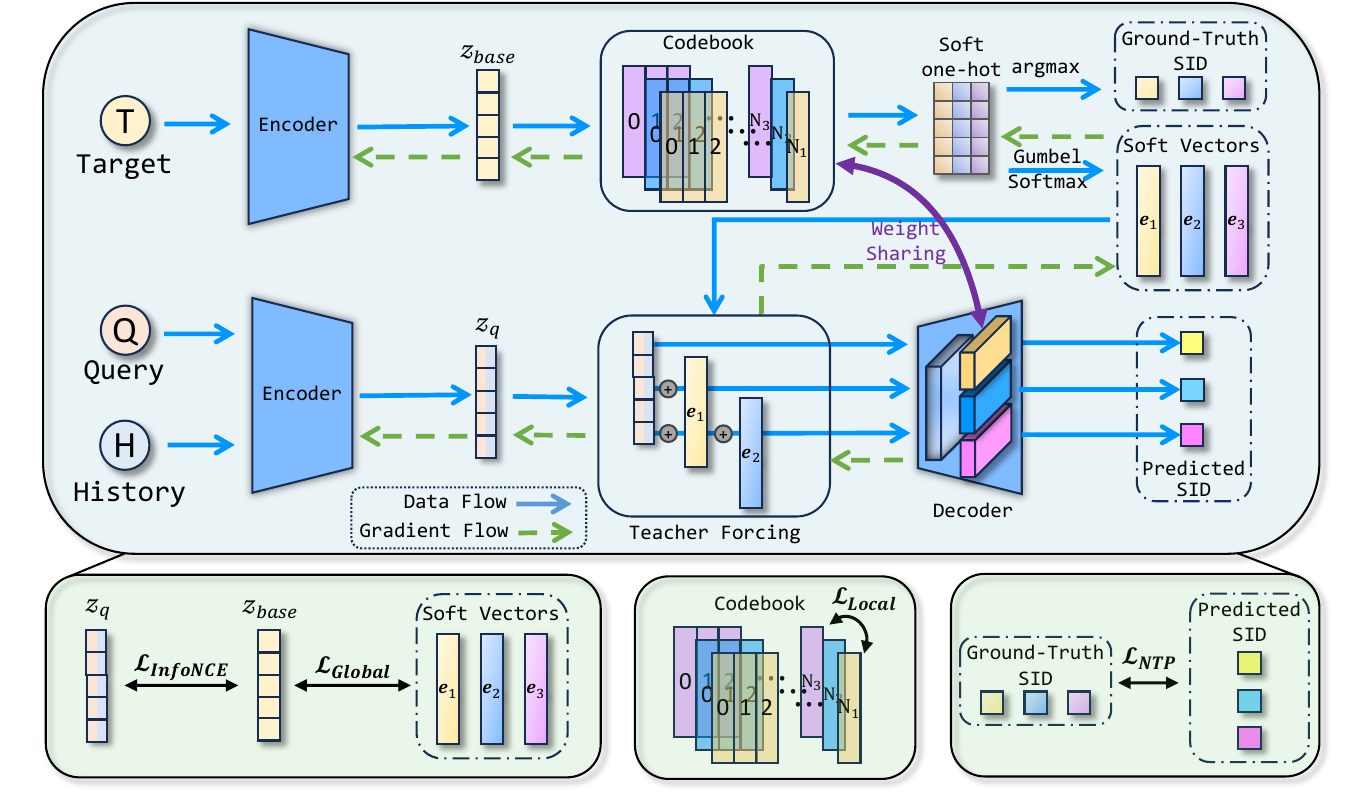}
  \caption{\textbf{Schematic Overview of the Differentiable Geometric Indexing (DGI) Framework.} 
Unlike two-stage methods that block optimization at the discrete indexing step, DGI establishes a fully differentiable pathway.
(1) Operational Unification: During training, we employ Gumbel-Softmax to generate \textit{soft quantized vectors}. These are fed directly into the decoder via Soft Teacher Forcing, enabling gradients (visualized as green dashed lines) to flow from the $\mathcal{L}_{NTP}$ loss back to the item encoder. We also enforce Symmetric Weight Sharing between the quantization codebooks and the decoder's prediction head to ensure a unified representation space.
(2) Geometric Optimization: The entire framework is optimized under spherical constraints (Scaled Cosine) to mitigate hubness.}
  \label{figm}
\end{figure*}

\section{Methodology}

Built upon the DGI framework defined in the Preliminaries, we now detail the specific instantiation of our End-to-End Generative Retrieval model. We first present the differentiable architecture that unifies indexing and retrieval, and then elaborate on the geometric optimization designed to resolve the Hubness problem.

\subsection{End-to-End Differentiable Architecture}
As illustrated in Figure \ref{figm}, DGI unifies indexing and retrieval into a single differentiable flow. 

\subsubsection{Operational Unification via Soft Gradient Flow}
Standard quantization employs the non-differentiable $\operatorname{argmax}$ operator, which necessitates stopping gradients during training. We replace this with a continuous relaxation to establish a valid gradient pathway.
At each quantization depth $j$, given the residual vector $\mathbf{r}_{j-1}$, we compute the selection probability over codebook $\mathbf{E}_j \in \mathbb{R}^{K \times d}$ using the Gumbel-Softmax reparameterization:
\begin{equation}
    p_{j,k} = \frac{\exp((s_{j,k} + g_{j,k}) / \tau_{GS})}{\sum_{l=1}^K \exp((s_{j,l} + g_{j,l}) / \tau_{GS})}
\end{equation}
where $s_{j,k}$ is the logit score, $g_{j,k} \sim \operatorname{Gumbel}(0, 1)$ is independent noise, and $\tau_{GS}$ is the temperature.
Crucially, instead of a discrete index, we compute a \textit{soft vector} $\tilde{\mathbf{z}}_j$ as the expectation:
\begin{equation}
    \tilde{\mathbf{z}}_j = \sum_{k=1}^K p_{j,k} \mathbf{e}_{j,k}, \quad \mathbf{r}_j = \mathbf{r}_{j-1} - \tilde{\mathbf{z}}_j
\end{equation}
By feeding these soft vectors $\tilde{\mathbf{z}}_j$ into the decoder during training, we enable \textbf{Soft Gradient Flow}. The gradient of the retrieval loss $\mathcal{L}_{NTP}$ with respect to the item encoder output $\mathbf{z}_{base}$ can now be computed via the chain rule:
\begin{equation}
    \frac{\partial \mathcal{L}_{NTP}}{\partial \mathbf{z}_{base}} = \sum_{j=1}^m \frac{\partial \mathcal{L}_{NTP}}{\partial \tilde{\mathbf{z}}_j} \cdot \mathbf{J}_{soft}^{(j)} \cdot \frac{\partial \mathbf{z}_{base}}{\partial \theta_{enc}}
\end{equation}
where $\mathbf{J}_{soft}$ is the non-zero Jacobian of the Gumbel-Softmax operation. This allows the retrieval objective to dynamically reshape the semantic space of the encoder.

\subsubsection{Representational Unification via Symmetric Weight Sharing}
To ensure the retrieval objective explicitly reshapes the index structure, we perform a critical architectural modification to the generator. 
Standard generative models utilize a monolithic language model head (lm\_head) that projects hidden states to a fixed vocabulary. This is structurally incompatible with our hierarchical semantic identifiers.

\textbf{Hierarchical Projection Heads.} 
We discard the pre-trained lm\_head and replace it with $M$ lightweight, parallel classification heads, corresponding to the $M$ layers of the RQ-VAE (where $M=3$ in our implementation). 
For the $m$-th layer of the semantic ID, the probability of predicting code $k$ at step $t$ is computed by projecting the decoder's hidden state $\mathbf{h}_t$ onto the codebook space.

\textbf{Symmetric Parameter Constraint.} 
Crucially, instead of initializing new parameters for these projection heads, we enforce strict Symmetric Weight Sharing. The weight matrix of the $m$-th classifier, denoted as $\mathbf{W}_{out}^{(m)}$, is explicitly defined as the transpose of the $m$-th quantization codebook $\mathbf{E}^{(m)}$:
\begin{equation}
    \mathbf{W}_{out}^{(m)} \equiv \mathbf{E}^{(m)}, \quad \text{where } \mathbf{E}^{(m)} \in \mathbb{R}^{K \times d}
\end{equation}
Consequently, the prediction logit for code $k$ in layer $m$ becomes a direct measure of geometric similarity between the decoder state and the code embedding:
\begin{equation}
    P(y_t^{(m)} = k | \mathbf{h}_t) = \text{Softmax}\left( \gamma \cdot \frac{\mathbf{h}_t \cdot (\mathbf{e}_k^{(m)})^\top}{\|\mathbf{h}_t\| \|\mathbf{e}_k^{(m)}\|} \right)
\end{equation}
This design ensures Semantic Alignment: the decoder does not learn a separate mapping to predict codes; instead, it learns to generate hidden states that directly align with the physical geometry of the index codebook. This eliminates the ``translation gap'' and forces the index and retriever to co-evolve within a shared semantic manifold.

\subsection{Isotropic Geometric Optimization}
\label{Isotropic}
While the end-to-end architecture allows for joint training, the quality of the learned index is further bounded by the geometry of the embedding space. We identify \textit{Geometric Anisotropy} as a critical bottleneck and propose a Riemannian geometric correction.

\subsubsection{The Hubness Problem}
\label{hubness}

Standard generative models typically compute logits via the unnormalized Dot Product: $s(q, i) = \mathbf{h}_q^\top \mathbf{e}_i$. 
In high-dimensional spaces, this metric introduces a critical pathology: Magnitude-Semantics Entanglement. We analyze why this leads to popularity-induced bias under power-law distributions.
Decomposing the dot product reveals that the score is determined by both geometric alignment and magnitude:
\begin{equation}
    s(q, i) = \|\mathbf{h}_q\| \|\mathbf{e}_i\| \cos(\theta_{q,i})
    \label{eq:dot_decomp}
\end{equation}
Consider the optimization dynamics for an item $i$ using the standard Cross-Entropy loss. For a positive query-item pair $(q, i)$, the gradient descent update on the item embedding $\mathbf{e}_i$ includes a term proportional to the query vector (see derivation in Appendix \ref{sec:appendix_proof}) :
\begin{equation}
    \mathbf{e}_i^{(t+1)} \leftarrow \mathbf{e}_i^{(t)} + \eta \cdot (1 - P(i|q)) \cdot \mathbf{z}_q
\end{equation}
where $\eta$ is the learning rate.

For high-frequency popular items (large $|\mathcal{D}_i^+|$), this positive update occurs significantly more often than for tail items. 
In many retrieval settings, queries associated with the same item are not isotropic but concentrate in related semantic regions, making $\mathbf{h}_q$ partially aligned on average.

Consequently, the summation of these coherent gradient updates leads to a net accumulation in magnitude, causing $\|\mathbf{e}_i\|$ to grow with the item's frequency, provided that no strict norm constraints (e.g., strong weight decay or normalization) are applied.

This popularity-driven norm growth can yield high logits even with only moderate angular alignment. Such high-norm items become hubs, appearing in the  top-$k$ results for disproportionately many queries and thereby suppressing long-tail items largely due to norm-dominated geometry.

\subsubsection{Scaled Cosine Geometry}
To resolve this, we enforce an Isotropic distribution by constraining all representations to the Unit Hypersphere $\mathbb{S}^{d-1}$. We replace the dot product with Scaled Cosine Similarity in both the Quantizer and the Decoder:
\begin{equation}
    P(c_l = k | \mathbf{z}_q, \mathbf{c}_{<l}) = \operatorname{Softmax}\left( \gamma \cdot \frac{\mathbf{h}_l}{\|\mathbf{h}_l\|_2} \cdot \frac{(\mathbf{e}_k^{(l)})^\top}{\|\mathbf{e}_k^{(l)}\|_2} \right)
\end{equation}

where $\gamma$ is a learnable scaling parameter (initialized to 30.0). This constraint ensures that retrieval probability depends \textit{exclusively} on the semantic angle, effectively ``flattening'' the popularity bias.

\subsubsection{Theoretical Insight: Riemannian Gradient Dynamics}
We rigorously justify our geometric design by showing its equivalence to optimization on a Riemannian manifold.
Let $\mathcal{L}$ be the loss function. In Euclidean space, the gradient update $\nabla \mathcal{L}(\mathbf{e})$ for an embedding $\mathbf{e}$ is unbounded. However, on the manifold $\mathbb{S}^{d-1}$, the true direction of steepest descent is given by the Riemannian Gradient $\operatorname{grad}\mathcal{L}(\mathbf{e})$, which projects the Euclidean gradient onto the tangent space $T_{\mathbf{e}}\mathbb{S}^{d-1}$:
\begin{equation}
    \operatorname{grad}\mathcal{L}(\mathbf{e}) := (\mathbf{I} - \mathbf{e}\mathbf{e}^\top)\nabla \mathcal{L}(\mathbf{e})
\end{equation}
Substituting our Scaled Cosine logit $s = \gamma \cos \theta$ (where $\cos \theta = \hat{\mathbf{h}}^\top \hat{\mathbf{e}}$), the update direction becomes:
\begin{equation}
    \operatorname{grad}\mathcal{L}(\mathbf{e}) \propto \frac{\gamma}{\|\mathbf{e}\|} (\hat{\mathbf{h}} - \cos \theta \cdot \hat{\mathbf{e}})
\end{equation}
Geometrically, the term $(\hat{\mathbf{h}} - \cos \theta \cdot \hat{\mathbf{e}})$ represents the component of the query vector $\hat{\mathbf{h}}$ that is \textit{orthogonal} to the item embedding $\hat{\mathbf{e}}$.
Unlike the Dot Product gradient $\nabla_{dot} \propto \hat{\mathbf{h}}$ which encourages norm inflation, our Riemannian update strictly removes the radial component ($\cos \theta \cdot \hat{\mathbf{e}}$) that affects magnitude. This mathematically guarantees that the gradient energy is dedicated solely to angular rotation, optimizing semantic alignment without introducing magnitude-induced noise (Hubness).
Our training process using F.normalize effectively approximates a Riemannian Stochastic Gradient Descent\cite{bonnabel2013stochastic} with Retraction, ensuring stable convergence on the isotropic manifold.

\subsection{Unified Training Objectives}
To facilitate the stability of the DGI framework, we employ a compound objective function that synergizes generative, reconstructive, and discriminative constraints.

\subsubsection{Generative and Reconstruction Tasks}
The primary objective is the Next Token Prediction (NTP) task, which drives the \textit{Operational Unification} via soft gradient flow. For a batch $\mathcal{B}$, we maximize the likelihood of the ground-truth identifier sequence $\mathbf{c}^{(i)}$:
\begin{equation}
    \mathcal{L}_{NTP} = - \frac{1}{|\mathcal{B}|} \sum_{(q, i) \in \mathcal{B}} \sum_{t=1}^m \log P(c_t^{(i)} | q, c_{<t}^{(i)}; \Theta)
\end{equation}

To ensure that the discrete codes retain semantic fidelity, we impose hierarchical reconstruction constraints. Consistent with our \textit{Isotropic Geometric Optimization}, we enforce Global Reconstruction using Cosine Distance rather than MSE, preventing norm-induced semantic drift:
\begin{equation}
    \mathcal{L}_{Global} = \frac{1}{|\mathcal{B}|} \sum_{i \in \mathcal{B}} \left( 1 - \frac{\mathbf{z}_{base}^{(i)} \cdot (\sum_{j=1}^m \tilde{\mathbf{z}}_j^{(i)})^\top}{\|\mathbf{z}_{base}^{(i)}\| \|\sum_{j=1}^m \tilde{\mathbf{z}}_j^{(i)}\|} \right)
\end{equation}
Simultaneously, we apply a standard Local Codebook loss to refine the quantization dictionary:
\begin{equation}
    \mathcal{L}_{Local} = \frac{1}{|\mathcal{B}|} \sum_{i \in \mathcal{B}} \sum_{j=1}^m \left( \|\text{sg}[\mathbf{r}_{j-1}] - \mathbf{e}_{j,k}\|^2 + \beta \|\mathbf{r}_{j-1} - \text{sg}[\mathbf{e}_{j,k}]\|^2 \right)
\end{equation}

\subsubsection{Alignment Strategy}
To stabilize the \textit{Representational Unification} (i.e., the shared parameter space), we maximize the mutual information between the query $\mathbf{z}_q$ and the target item $\mathbf{z}_{base}$ using in-batch negatives. This discriminative loss pushes the query representation closer to the target on the hypersphere:
    \begin{equation}
        \mathcal{L}_{InfoNCE} = - \log \frac{\exp(\hat{\mathbf{z}}_q^\top \hat{\mathbf{z}}_{base}^{i^+} / \tau_{CL})}{\sum_{n \in \mathcal{B}} \exp(\hat{\mathbf{z}}_q^\top \hat{\mathbf{z}}_{base}^{n} / \tau_{CL})}
    \end{equation}

\subsubsection{Diversity Regularization}
Finally, to prevent codebook collapse—a common failure mode in discrete representation learning—we maximize the entropy of the code usage distribution:
\begin{equation}
    \mathcal{L}_{Div} = \sum_{j=1}^m \sum_{k=1}^K \bar{p}_{j,k} \log (\bar{p}_{j,k} + \epsilon)
\end{equation}
The total objective is a weighted sum: $\mathcal{L}_{Total} = \mathcal{L}_{NTP} +  \mathcal{L}_{Global} +  \mathcal{L}_{Local} +  \mathcal{L}_{InfoNCE} +   \mathcal{L}_{Div}$.

\begin{table}[t]
\centering
\caption{Statistics of AOL4PS and AE-PV Datasets. ``Hist''. denotes interaction history availability. Lengths are measured in tokens.}
\label{tab:dataset_stats}
\begin{tabular}{l||cc||cc}
\toprule
\multirow{2}{*}{\textbf{Metric}} & \multicolumn{2}{c|}{\textbf{AOL4PS}} & \multicolumn{2}{c}{\textbf{AE-PV}} \\
\cmidrule(lr){2-3} \cmidrule(lr){4-5}
 & \textbf{Train} & \textbf{Test} & \textbf{Train} & \textbf{Test} \\
\midrule
\textbf{Scale} & & & & \\
Total Samples & 1,061K & 265K & 12.0M & 13.1M \\
\midrule
\textbf{Distribution} & & & & \\
w/ History (\%) & 100.0\% & 100.0\% & 73.6\% & 73.5\% \\
Sparsity & 0.00 & 0.00 & 0.26 & 0.26 \\
\midrule
\textbf{Avg. Length} & & & & \\
Avg. Hist. Len. & 9.45 & 10.00 & 6.26 & 6.23 \\
Avg. Query Len. & 2.00 & 1.99 & 3.16 & 3.15 \\
Avg. Target Len. & 6.92 & 6.86 & 18.54 & 18.55 \\
\bottomrule
\end{tabular}
\end{table}
\section{Experiments}

In this section, we provide empirical validation of the proposed DGI framework. Our experiments are designed to answer four core research questions:
\begin{itemize}
    \item \textbf{RQ1 (Overall Effectiveness):} Does DGI outperform state-of-the-art sparse, dense, and generative baselines?
    \item \textbf{RQ2 (Ablation Analysis):} How do the ``End-to-End Structural Unification'' and ``Isotropic Geometric Optimization'' individually contribute to performance?
    \item \textbf{RQ3 (Mechanism Verification):} Can we observe the restoration of gradient stability and geometric isotropy, particularly for long-tail items?
    \item \textbf{RQ4 (Online Evaluation):} Does DGI demonstrate practical scalability and robustness in a real-world online production environment, delivering significant gains in commercial metrics?
\end{itemize}

\begin{table*}[t]
\caption{Overall Performance Comparison. The best result for each metric is in \textbf{bold}, and the second-best is \underline{underlined}.}
\label{tab:main_results}
\centering
\small
\begin{tabular}{lccccccc||ccccccc}
\toprule
 \multicolumn{1}{c}{} & \multicolumn{7}{c||}{\textbf{AOL4PS}} & \multicolumn{7}{c}{\textbf{AE-PV}} \\
\cmidrule(lr){2-8} \cmidrule(lr){9-15} \\
 \textbf{Model} & H@1 & H@5 & H@10 &H@20 & N@5 & N@10 & N@20 & H@1 & H@5 & H@10 &H@20 & N@5 & N@10 & N@20 \\
\midrule
BM25 & 0.2833 & 0.3872  & 0.4177 & 0.4439 &  0.3399 & 0.3499 & 0.3565 & 0.0125 & 0.0363 &0.0540 & 0.0789 &0.0245 & 0.0303 &0.0365 \\
DocT5Query &0.1608 &0.2942 &0.3289 &0.3557 &0.2345 &0.2459 &0.2527 &0.0046 &0.0211 &0.0334 &0.0523 &0.0129 &0.0168 &0.0216\\
\midrule
DSSM-T5 & 0.4049 & \underline{0.6435} &\underline{0.7307} &\textbf{0.8022} &0.5331 & 0.5615 & 0.5797 & 0.0095 &0.0364 &0.0615 & 0.1005 &0.0229 &0.0310 &0.0408 \\
Sentence-T5 & \underline{0.4657} & 0.5352 & 0.6159 & 0.7268 & \underline{0.5568} & \underline{0.5902} & \underline{0.6061} & \textbf{0.0223} &\underline{0.0751} &\underline{0.1162} &\underline{0.1724} &\underline{0.0489} &\underline{0.0621} &\underline{0.0762} \\
\midrule
DSI & 0.1385 & 0.2485 & 0.2629 & 0.2788 & 0.2042 & 0.2089 & 0.2129 & 0.0037 & 0.0149 & 0.0258 & 0.0437 & 0.0092 & 0.0127 & 0.0172 \\ 
Two-Stage & 0.4323 & 0.4782  & 0.4952  & 0.5133 & 0.4570 & 0.4626 & 0.4671 & 0.0182 & 0.0301 & 0.0303 & 0.0314  & 0.0250 & 0.0263 & 0.0351 \\
UniSearch &0.3383 &0.4838 &0.5134 &0.5342 &0.4180 &0.4277 &0.4330 &0.0108 & 0.0397 &0.0652&0.1012 &0.0252 &0.0334 &0.0425\\
\midrule
\textbf{DGI (Ours)} & \textbf{0.5052} & \textbf{0.7247} & \textbf{0.7637} & \underline{0.7887} & \textbf{0.6265} & \textbf{0.6393} & \textbf{0.6457} & \underline{0.0216} & \textbf{0.0808} & \textbf{0.1315} & \textbf{0.2001} & \textbf{0.0512} & \textbf{0.0675} & \textbf{0.0848}\\
\bottomrule
\end{tabular}
\end{table*}

\begin{table*}[t]
\caption{\textbf{Comprehensive Ablation Study.} We decouple the contributions of Operational Unification and Geometric Restoration. "$\checkmark$" denotes the component is enabled, "$\times$" denotes disabled/replaced with standard baseline.}
\label{tab:ablation}
\small
\centering
\setlength{\tabcolsep}{4.5pt}
\begin{tabular}{l|ccc|cccc|cccc}
\toprule
\multirow{2}{*}{\textbf{Model Variant}} & \multicolumn{3}{c|}{\textbf{Key Components}} & \multicolumn{4}{c|}{\textbf{AOL4PS}} & \multicolumn{4}{c}{\textbf{AE-PV}} \\
 & \textbf{Soft} & \textbf{Weight} & \textbf{Scaled} & \multirow{2}{*}{\textbf{H@1}} & \multirow{2}{*}{\textbf{H@10}} & \multirow{2}{*}{\textbf{H@20}} & \multirow{2}{*}{\textbf{N@20}} & \multirow{2}{*}{\textbf{H@1}} & \multirow{2}{*}{\textbf{H@10}} & \multirow{2}{*}{\textbf{H@20}} & \multirow{2}{*}{\textbf{N@20}} \\
 & \textbf{Grad.} & \textbf{Share} & \textbf{Cosine} & & & & & & & & \\
\midrule
\textbf{DGI (Full Model)} & $\checkmark$ & $\checkmark$ & $\checkmark$ & \textbf{0.5651} & \textbf{0.7752} & \textbf{0.7980} & \textbf{0.6812} & \textbf{0.0149} & \textbf{0.0915} & \textbf{0.1417} & \textbf{0.0594} \\
\midrule
\multicolumn{12}{l}{\textit{Ablation of Operational Unification}} \\
w/o Soft Gradient Flow & $\times$ & $\checkmark$ & $\checkmark$ & 0.5216 & 0.7513 & 0.7762 & 0.6501 & 0.0101 & 0.0524 & 0.0813 & 0.0302 \\
w/o Weight Sharing & $\checkmark$ & $\times$ & $\checkmark$ & 0.5282 & 0.7561 & 0.7804 & 0.6603 & 0.099 & 0.0573 & 0.0896 & 0.0311 \\
w/o Both (Fully Decoupled) & $\times$ & $\times$ & $\checkmark$ & 0.5193 & 0.7504 & 0.7709 & 0.6485 & 0.0097 & 0.0426 & 0.0606 & 0.0284 \\
\midrule
\multicolumn{12}{l}{\textit{Ablation of Isotropic Geometric Optimization}} \\
w/o Scaled Cosine & $\checkmark$ & $\checkmark$ & $\times$ & 0.3768 & 0.5242 & 0.5445 & 0.4573 & 0.0108 & 0.0652 & 0.1012 & 0.0425 \\
\midrule
\multicolumn{12}{l}{\textit{Naive Baseline}} \\
w/o All Components & $\times$ & $\times$ & $\times$ & 0.0531 & 0.0859 & 0.0969 & 0.1171 & 0.0002 & 0.0004 & 0.0005 & 0.0005 \\
\bottomrule
\end{tabular}
\end{table*}

\begin{figure}[hbpt]
\centering
\begin{subfigure}[t]{0.48\columnwidth}
  \centering
  \includegraphics[width=\linewidth]{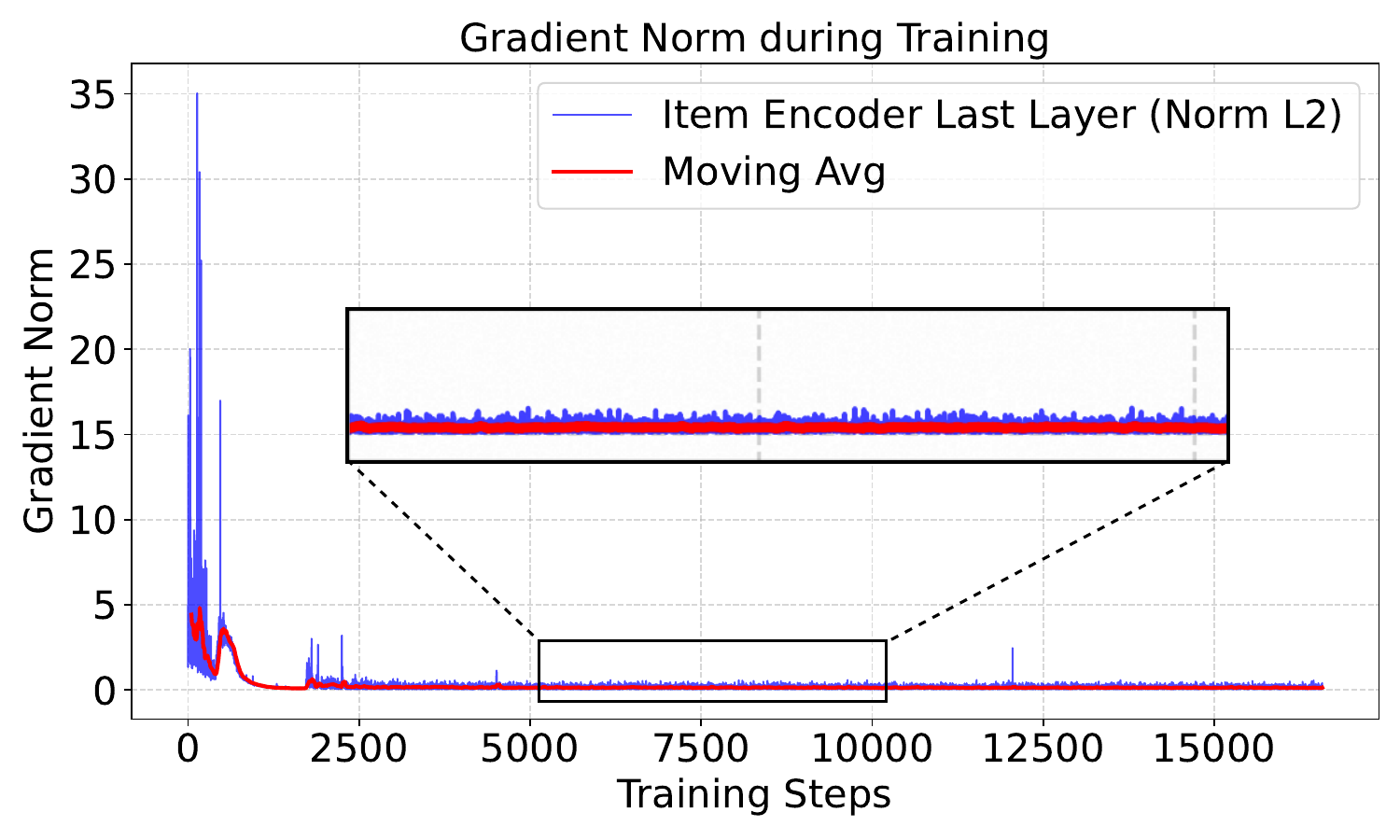}
  \caption{\textbf{DGI (Soft Gradient)}}
\end{subfigure}\hfill
\begin{subfigure}[t]{0.48\columnwidth}
  \centering
  \includegraphics[width=\linewidth]{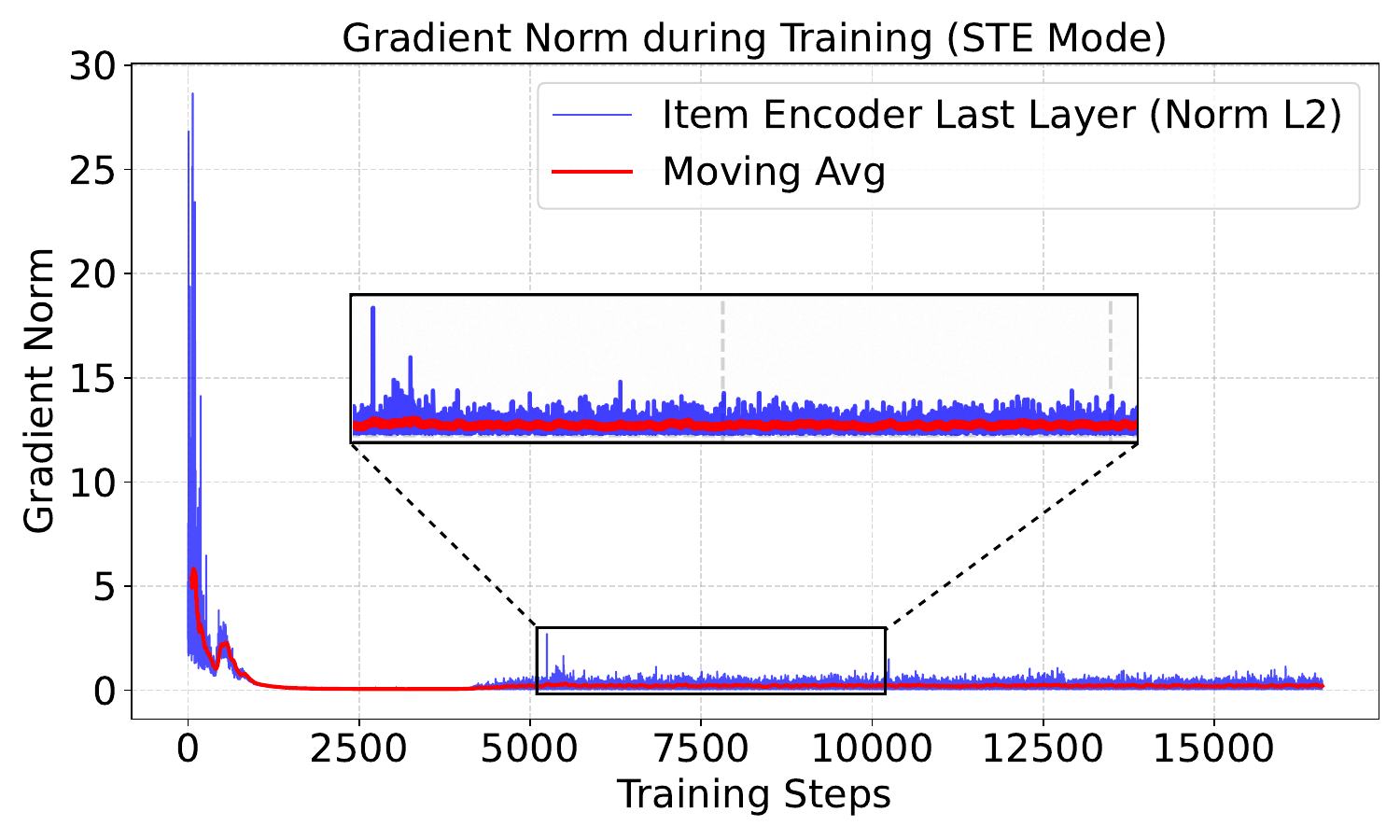}
  \caption{\textbf{Baseline (STE Variant)}}
\end{subfigure}
\caption{\textbf{Analysis of Optimization Stability (Gradient Norms).} 
Visualization of gradient magnitude variance on the same dataset. 
(a) \textbf{DGI} exhibits smooth and consistent gradient flows, indicating that our Soft Teacher Forcing effectively stabilizes the backward pass.
(b) The \textbf{STE Baseline} (replacing Gumbel-Softmax with STE) suffers from severe oscillation and high-variance spikes. 
This contrast empirically validates that the differentiable pathway is crucial for stable end-to-end index learning.}
\label{fig3}
\end{figure}

\begin{figure}[hbpt]
\centering
\begin{subfigure}[t]{0.49\columnwidth}
  \centering
  \includegraphics[width=\linewidth]{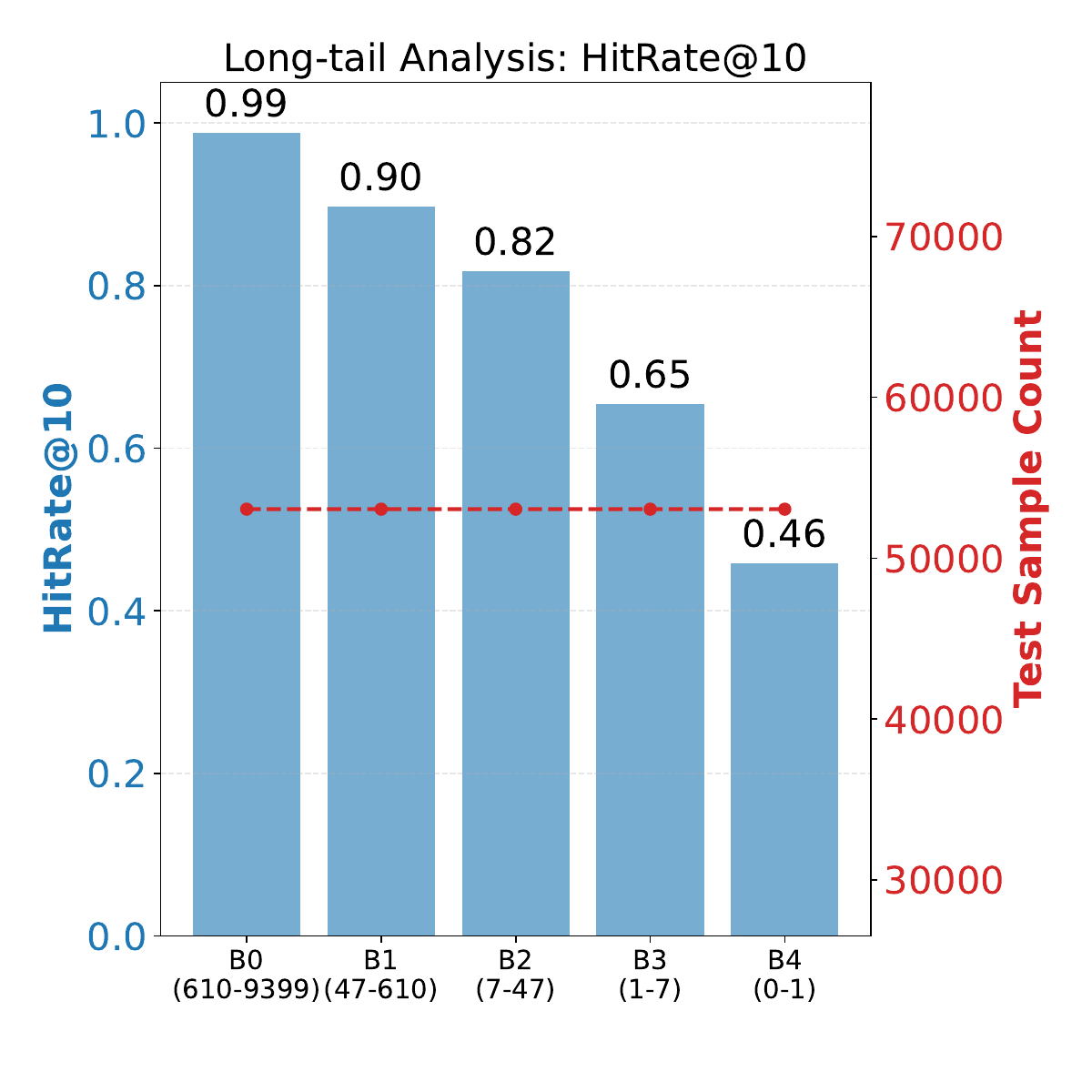}
  \caption{\textbf{DGI}}
\end{subfigure}\hfill
\begin{subfigure}[t]{0.49\columnwidth}
  \centering
  \includegraphics[width=\linewidth]{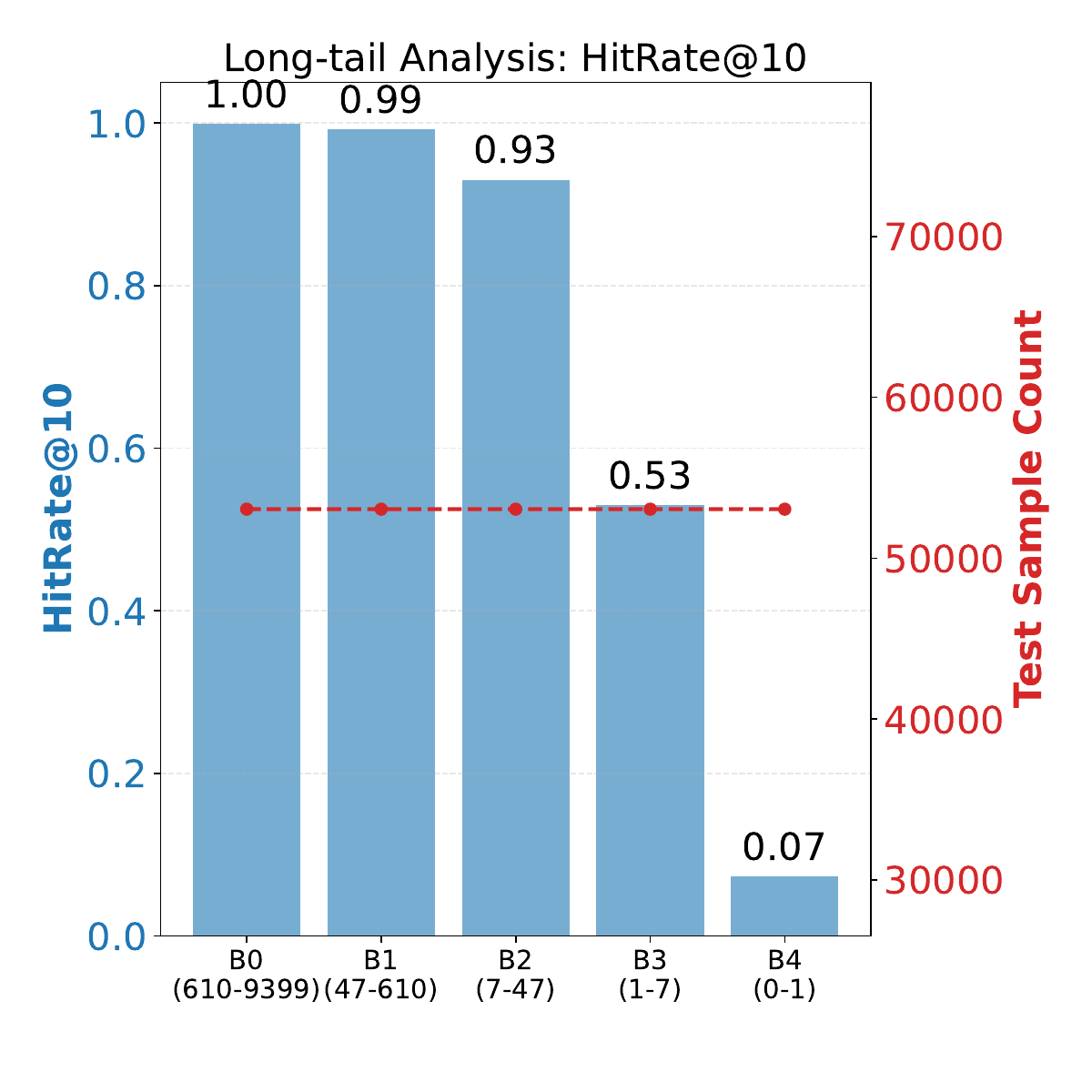}
  \caption{\textbf{Two-Stage with Dot Product}}
\end{subfigure}
\caption{\textbf{Long-tail Robustness Analysis.} Performance (HitRate) across item popularity deciles (B0: Head $\to$ B4: Tail).
(a) \textbf{DGI} maintains a robust and uniform performance profile across all buckets, demonstrating that our Isotropic Geometric Optimization effectively recalls long-tail items.
(b) In contrast, the \textbf{Two-Stage with Dot Product} suffers from a classic ``rich-get-richer'' pattern, where performance collapses in the tail buckets (B3-B4) due to severe popularity bias.}
\label{fig4}
\end{figure}

\subsection{Experimental Setup}

\subsubsection{Datasets}
We evaluate our proposed DGI framework on two large-scale datasets covering web search and e-commerce search. Detailed statistics for both datasets are summarized in Table \ref{tab:dataset_stats}.

\textbf{AOL4PS \cite{guo2021aol4ps}:} Derived from the AOL web query logs, this dataset is tailored for personalized search. We reconstruct the textual sessions by mapping query and document indices back to their raw text content. For each user interaction, we construct a sample input consisting of the current query text and a history window of the preceding 10 clicked document titles, with the current clicked title as the prediction target. To simulate a realistic streaming retrieval scenario and prevent future information leakage, we strictly adopt a global time-based split. All samples are sorted chronologically by timestamp; the first 80\% constitute the training set, while the subsequent 20\% form the testing set. Note that this split is global rather than user-wise, allowing a user's early behaviors to appear in the training set while their later behaviors appear in the test set.

\textbf{AE-PV:} A proprietary dataset collected from the real-world Page View (PV) logs of a leading e-commerce platform. It is characterized by high sparsity and strict semantic matching requirements. Each sample includes a search query and a history sequence of the user's last 10 clicked product titles. The dataset spans two consecutive days of traffic: we utilize the data from the first day for training and the second day for testing. This setup rigorously evaluates the model's ability to generalize to future queries based on past interactions.

\subsubsection{Baselines}
We compare DGI against baselines from three paradigms. To ensure fairness, all methods are adapted to be context-aware using both query and history (details in Appendix \ref{app:baseline}) :

\textbf{Sparse Retrieval:} We include \textbf{BM25} \cite{robertson2009probabilistic} as the standard for exact lexical matching, and \textbf{DocT5Query} \cite{nogueira2019doc2query}, which enhances BM25 by expanding documents with generated queries to mitigate vocabulary mismatch.
    
\textbf{Dense Retrieval:} We implement the classic \textbf{DSSM} \cite{huang2013learning} using the T5 encoder as the backbone to ensure a fair comparison of model capacity. We also include \textbf{Sentence-T5} \cite{ni2022sentence}, a strong baseline that fine-tunes pre-trained T5 encoders with contrastive loss.
    
\textbf{Generative Retrieval:} We compare against \textbf{DSI} \cite{tay2022transformer} and the \textbf{Two-Stage baseline}, which is implemented based on TIGER \cite{tiger2023} by incorporating a query encoder to train a sequence-to-sequence retriever that maps user queries to the fixed item codes established by the RQ-VAE. Additionally, we compare against \textbf{UniSearch} \cite{chen2025unisearch}, a representative joint training framework. 

\subsubsection{Metrics}
We report \textbf{HitRate@K} (H@K) and \textbf{NDCG@K}(N@K) for $K\in\{1, 5, 10, 20\}$. HitRate measures recall capability, while NDCG evaluates the ranking quality.
\subsubsection{Implementation Details}

We implemented DGI on a single NVIDIA A100 GPU.

\textbf{Architecture \& Indexing.} We utilize the pre-trained T5-large as the backbone. The hierarchical indexer employs Residual Quantization (RQ-VAE) with a depth of $m=3$ and hierarchical codebook sizes of $\{4096, 1024, 256\}$. We initialize the codebooks via K-means clustering. 
Furthermore, the indexer strictly enforces the Scaled Cosine Geometry: both the input residuals and codebook embeddings are $\ell_2$-normalized and scaled before code assignment, ensuring isotropic optimization on the hypersphere. 

\textbf{Training Strategy.} We use the AdamW optimizer with a differential learning rate strategy: $5 \times 10^{-5}$ for the pre-trained backbone and $5 \times 10^{-4}$ for the quantizer to accelerate codebook convergence. We apply a linear warmup for the first 1,000 steps followed by linear decay. The global batch size is set to 128.

\textbf{Inference.} We employ a Coarse-to-Fine strategy. Offline, we map quantized SIDs to an inverted index and cache normalized item embeddings. Online, we identify candidate semantic clusters via Constrained Beam Search. Candidates within these clusters are then ranked using a Lexicographical Rule: sorting primarily by the generative beam score, and secondarily by the fine-grained geometric similarity ($\hat{\mathbf{z}}_q^\top \hat{\mathbf{z}}_{item}$) to resolve hash collisions.

\textbf{Hyperparameters.} To stabilize the Scaled Cosine objective, the scaling factor $\gamma$ is explicitly initialized to $30.0$. For Gumbel-Softmax relaxation, the temperature $\tau$ is annealed exponentially from $1.0$ to $0.1$. During inference, we generate Semantic IDs using beam search with a beam size from 1 to 20.

\subsection{Overall Performance (RQ1)}
As summarized in Table \ref{tab:main_results}, DGI consistently outperforms representative sparse, dense, and generative baselines across both datasets. Notably, DGI demonstrates a decisive advantage over existing generative paradigms, achieving a \textbf{4.3x} improvement in HitRate@10 on the challenging AE-PV dataset compared to the Two-Stage baseline and significantly surpassing the joint-training method UniSearch. Furthermore, DGI excels in ranking quality against strong dense retrievers like Sentence-T5 (e.g., +8.3\% NDCG@10 on AOL4PS), validating that our isotropic geometric optimization yields superior semantic fidelity compared to standard contrastive dual-encoders.

\subsection{Ablation Studies (RQ2)}

Table \ref{tab:ablation} provides a rigorous component-wise validation of the DGI framework. First, Operational Unification proves foundational; severing the optimization link (w/o Soft Gradient) or the parameter alignment (w/o Weight Sharing) hinders the back-propagation of retrieval signals, leading to consistent performance degradation (e.g., H@1 drops from 0.5651 to 0.5216 on AOL4PS). Second, Geometric Restoration is identified as the critical safeguard against popularity bias; reverting to the unnormalized Dot Product (w/o Scaled Cosine) triggers a precipitous 33.3\% collapse in H@1 (0.5651 → 0.3768), confirming that norm inflation severely distorts semantic ranking even within a differentiable architecture. Ultimately, removing all components degrades the system to a near-random baseline (0.0531), demonstrating that structural differentiability and geometric isotropy are not merely additive enhancements, but mutually dependent pillars essential for stable end-to-end learning.

\subsection{In-depth Mechanism Analysis (RQ3)}

\subsubsection{Optimization Stability: Gradient Norm Analysis}
Figure \ref{fig3} visualizes the gradient norm dynamics of the item encoder. The baseline (STE) suffers from severe oscillation and high-variance spikes, confirming that the surrogate estimator introduces significant optimization noise that destabilizes convergence. In sharp contrast, DGI exhibits smooth and consistent gradient flows throughout training. This empirically verifies that our \textit{Soft Gradient Flow} effectively establishes a low-variance, differentiable pathway, allowing the retrieval objective to reliably backpropagate and refine the continuous representation space without the instability inherent in discrete approximations.

\subsubsection{Geometric Isotropy: Long-tail Robustness}
Figure \ref{fig4} validates the efficacy of our \textit{Scaled Cosine Geometry} across AOL4PS datasets item popularity deciles. The Two-Stage baseline suffers from a steep performance collapse in tail buckets, exhibiting a classic ``rich-get-richer'' pattern driven by unconstrained norm inflation. In contrast, DGI maintains a remarkably uniform performance profile with robust accuracy even on zero-shot items. This empirically confirms that our isotropic spherical constraint effectively decouples popularity magnitude from semantic relevance, rendering long-tail items geometrically accessible regardless of their interaction frequency.

\subsubsection{Topological Visualization: t-SNE}
Figure \ref{fig5} provides a topological visualization of the learned semantic space. The Two-Stage baseline exhibits severe \textit{Representation Collapse} (highlighted by red arrows), where embeddings crowd into a narrow, anisotropic cone, leaving the latent space largely underutilized. In contrast, DGI forms a highly \textit{isotropic} distribution with distinct, well-separated clusters spanning the entire unit hypersphere. This visual evidence confirms that our Riemannian optimization effectively maximizes the entropy of the representation space, ensuring that the index capacity is fully utilized to capture diverse item semantics.

\begin{figure}[hbpt]
\centering
\begin{subfigure}[t]{0.49\columnwidth}
  \centering
  \includegraphics[width=\linewidth]{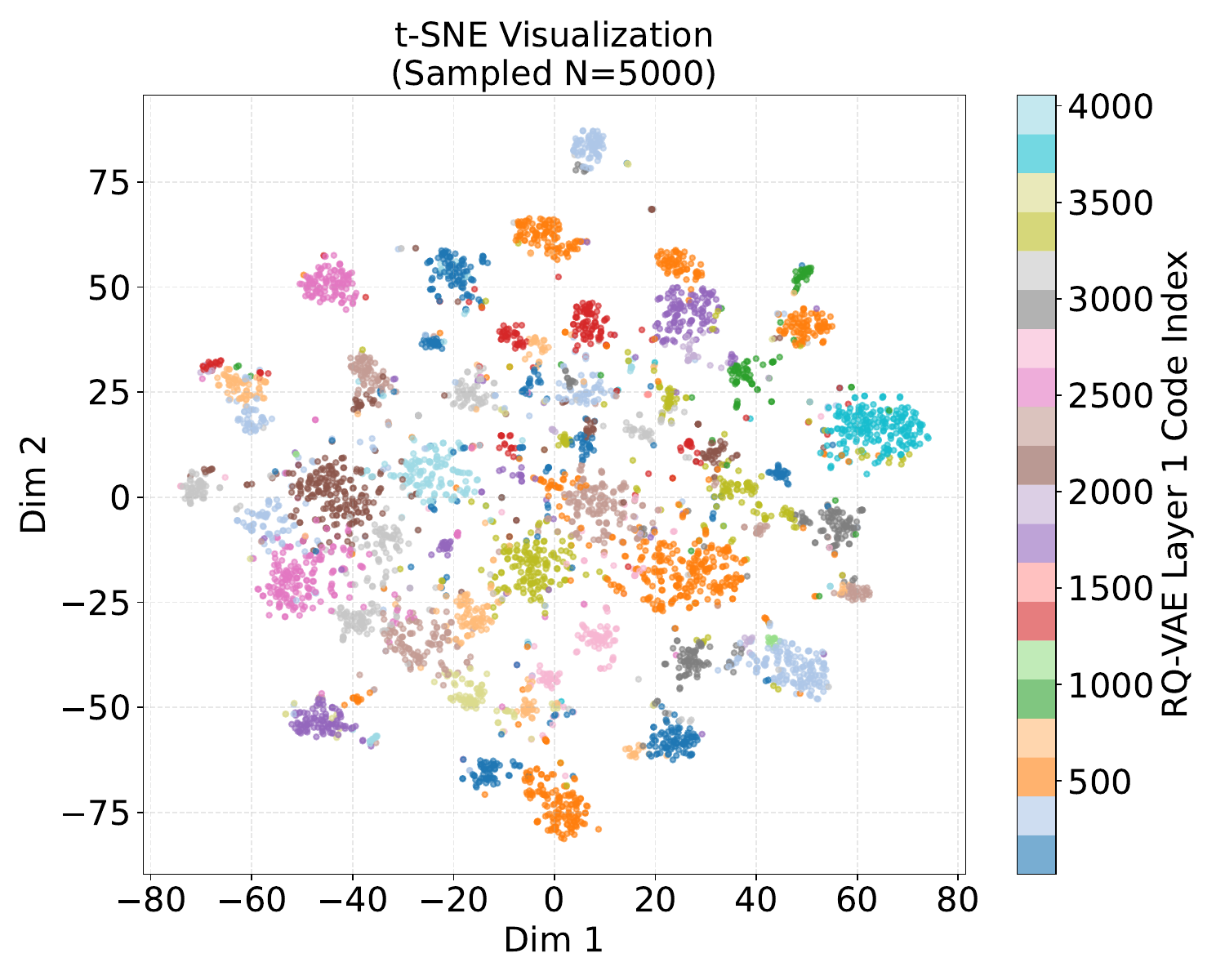}
  \caption{DGI}
\end{subfigure}\hfill
\begin{subfigure}[t]{0.49\columnwidth}
  \centering
  \includegraphics[width=\linewidth]{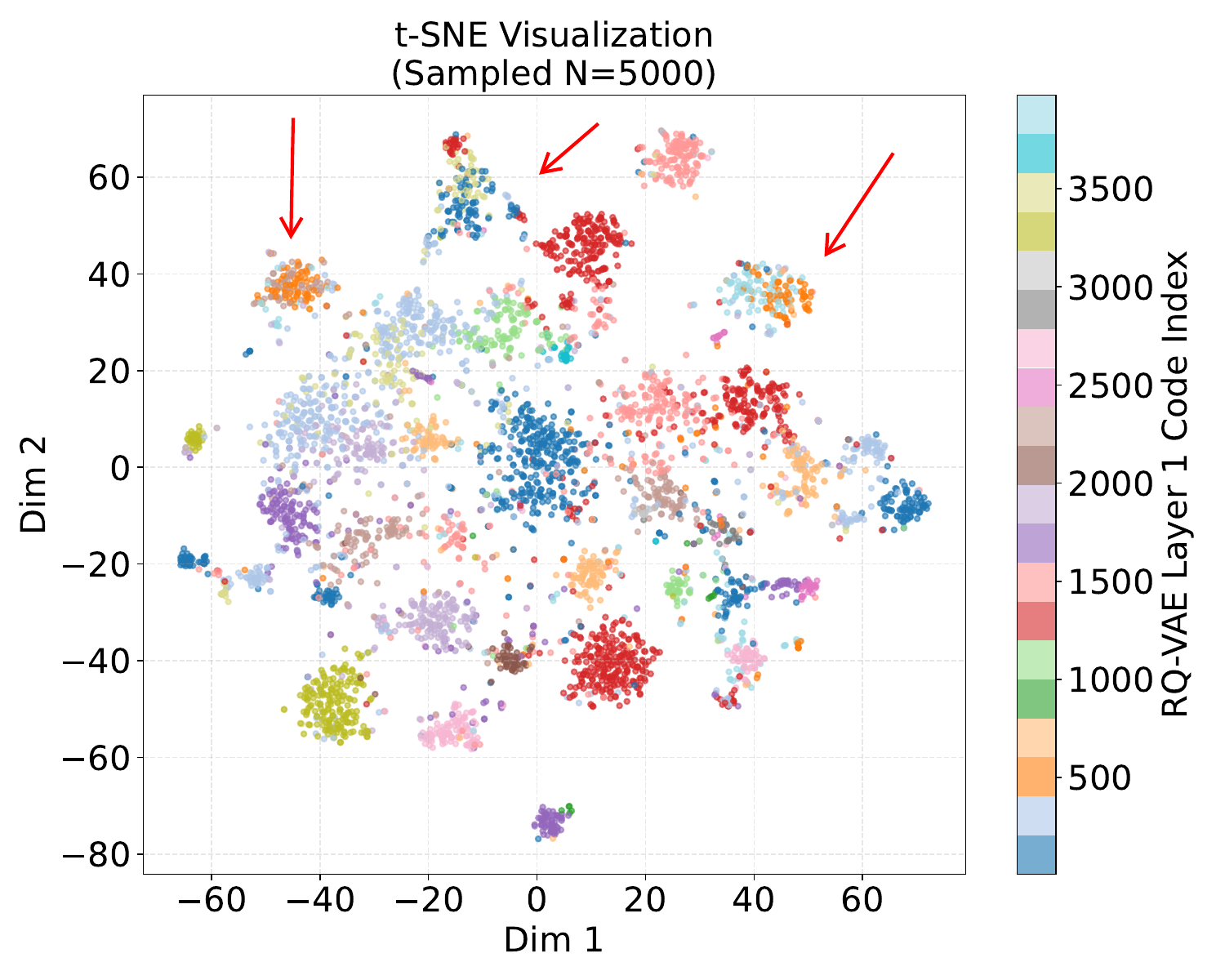}
  \caption{Two-Stage}
\end{subfigure}
\caption{\textbf{Topological Visualization of Learned Semantic Spaces (t-SNE).} (a) DGI learns a highly \textit{Isotropic} distribution on the hypersphere with well-separated clusters, confirming the geometric restoration capability of our Riemannian optimization. (b) The baseline space exhibits \textit{Representation Collapse} and \textit{Anisotropy} (indicated by red arrows), where embeddings crowd into a narrow cone, leaving the semantic space underutilized.}
\label{fig5}
\end{figure}
\subsection{Online Evaluation (RQ4)}

We conducted a 7-day online A/B test on a leading global e-commerce platform serving hundreds of millions of users. The experiment compared our generative model against a control group running the production hybrid system (integrating lexical and embedding modules consistent with our offline baselines). To ensure low latency, we employed a decoupled inference strategy, using the query's top-6 clicked items as a history proxy for training consistency. As a supplementary recall channel, our model yielded statistically significant gains of +1.27\% CTR and +1.11\% RPM (both $p < 0.001$), validating its effectiveness in a large-scale industrial setting.

\section{Conclusion}

In this work, we identify two intrinsic bottlenecks in Generative Retrieval: the \textit{Optimization Gap} due to non-differentiable indexing and the \textit{Geometric Conflict} causing norm-dominated Hubness. 
We propose Differentiable Geometric Indexing (DGI) to resolve these via \textit{Operational Unification} and \textit{Isotropic Geometric Optimization}. 
By establishing a fully differentiable pathway and enforcing spherical constraints, DGI transforms the index from a static artifact into a dynamic structure that co-evolves with retrieval intent. 
Theoretical analysis and extensive experiments---including a significant online A/B test---confirm that DGI effectively eliminates popularity bias and outperforms state-of-the-art baselines, providing a robust paradigm for next-generation industrial retrieval.

\bibliographystyle{ACM-Reference-Format}
\bibliography{ref}

\appendix

\section{Theoretical Analysis}
\label{sec:appendix_proof}

In this section, we provide the convergence analysis for the proposed optimization framework. We clarify that our implementation formally employs Weight Normalization (Reparameterization): we optimize unconstrained parameters $\mathbf{w} \in \mathbb{R}^d$ using standard optimizers, while the model utilizes the normalized projection $\theta(\mathbf{w}) = \mathbf{w} / \|\mathbf{w}\|_2 \in \mathbb{S}^{d-1}$ during the forward pass.

\subsection{Problem Formulation}
We consider the \textbf{continuous soft surrogate objective} $J(\mathbf{w})$ defined by the Gumbel-Softmax relaxation (with `hard=False`). The optimization problem is:
\begin{equation}
    \min_{\mathbf{w} \in \mathbb{R}^d \setminus \{\mathbf{0}\}} J(\mathbf{w}) := \mathbb{E}_{\xi} [\mathcal{L}(\theta(\mathbf{w}); \xi)]
\end{equation}
where $\xi$ represents the stochasticity from data sampling and Gumbel noise.

\begin{figure}
    \centering
    \includegraphics[width=0.5\linewidth]{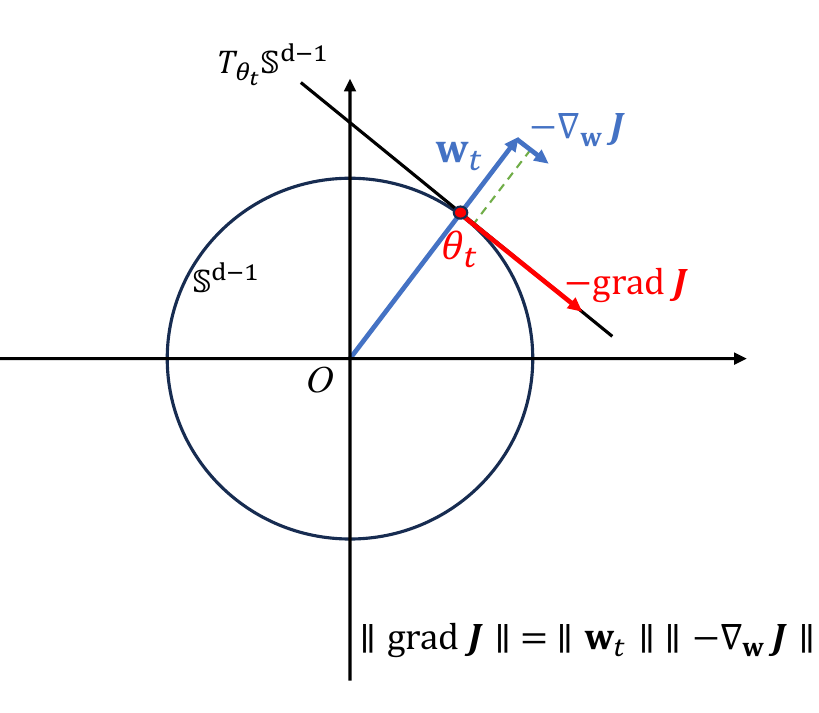}
    \caption{Geometric interpretation of Weight Normalization. The optimization is performed on the unconstrained parameter $\mathbf{w}$ in Euclidean space. The Euclidean gradient $\nabla_{\mathbf{w}} J$ is orthogonal to the radial direction and is strictly parallel to the Riemannian gradient $\operatorname{grad} J$ on the unit sphere $\mathbb{S}^{d-1}$, scaled by $1/\|\mathbf{w}\|_2$.}
    \label{figa}
\end{figure}

\subsection{Gradient Dynamics}
The relationship between the Euclidean gradient w.r.t. $\mathbf{w}$ ($\nabla_{\mathbf{w}} J$) and the Riemannian gradient w.r.t. $\theta$ ($\operatorname{grad} J$) is derived via the chain rule:
\begin{equation}
    \label{eq:grad_relation}
    \nabla_{\mathbf{w}} J(\mathbf{w}) = \frac{1}{\|\mathbf{w}\|_2} (\mathbf{I} - \theta\theta^\top) \nabla_{\theta} J(\theta) = \frac{1}{\|\mathbf{w}\|_2} \operatorname{grad}_{\mathbb{S}^{d-1}} J(\theta)
\end{equation}
This implies the norm relationship:
\begin{equation}
    \label{eq:norm_relation}
    \|\operatorname{grad}_{\mathbb{S}^{d-1}} J(\theta)\|_2 = \|\mathbf{w}\|_2 \cdot \|\nabla_{\mathbf{w}} J(\mathbf{w})\|_2
\end{equation}

\subsection{Convergence Analysis}
We analyze the convergence of the underlying Stochastic Gradient Descent (SGD) process on $\mathbf{w}$. Let $\mathcal{F}_t$ denote the filtration generated by the random variables (initialization and noise sequence) up to iteration $t$, such that the parameter $\mathbf{w}_t$ is $\mathcal{F}_t$-measurable, but the noise $\xi_t$ used to compute the gradient $g_t$ is independent of $\mathcal{F}_t$.

\textbf{Assumptions:}
\begin{itemize}
    \item \textbf{A1 (Bounded Norm):} We assume the trajectory of the parameter norm remains within a bounded ring domain, i.e., there exist constants $0 < m \le M < \infty$ such that $m \le \|\mathbf{w}_t\|_2 \le M$ for all $t$.
    \textit{Justification:} This is a technical assumption to ensure smoothness. In training, this can be rigorously enforced by explicit norm clipping or regularization if necessary.
    
    \item \textbf{A2 (Smoothness):} The objective function $J(\mathbf{w})$ is $L$-smooth on the active domain $\mathcal{D} = \{\mathbf{w} : m \le \|\mathbf{w}\| \le M\}$.
    \textit{Justification:} We assume the objective behaves smoothly within the bounded domain $\mathcal{D}$, analyzing a smoothed approximation of the neural network if non-smooth activations are present.
    
    \item \textbf{A3 (Unbiased Soft Estimator):} The stochastic gradient $g_t$ is an unbiased estimator of the \textit{soft surrogate} gradient $\nabla_{\mathbf{w}} J(\mathbf{w}_t)$ with bounded variance: 
    \begin{equation}
        \mathbb{E}[g_t | \mathcal{F}_t] = \nabla_{\mathbf{w}} J(\mathbf{w}_t), \quad \mathbb{E}[\|g_t - \nabla_{\mathbf{w}} J(\mathbf{w}_t)\|^2 | \mathcal{F}_t] \le \sigma^2
    \end{equation}
    
    \item \textbf{Remark on Discretization:} The theorem below strictly applies to the soft relaxation. In our experiments, hard SIDs are used only as supervision targets (argmax), while the teacher-forcing conditioning path uses the soft relaxation; therefore gradients propagate through the relaxation.

    \item \textbf{A4 (Lower Bound):} The objective function is bounded below, i.e., $J(\mathbf{w}) \ge J^* > -\infty$.
\end{itemize}

\begin{theorem}
Under Assumptions A1-A4, if the learning rate sequence satisfies the Robbins-Monro conditions ($\sum \eta_t = \infty, \sum \eta_t^2 < \infty$), applying SGD on $\mathbf{w}$ guarantees that the Riemannian gradient norm vanishes asymptotically in the sense of limit inferior:
\begin{equation}
    \liminf_{t \to \infty} \|\operatorname{grad}_{\mathbb{S}^{d-1}} J(\theta_t)\| = 0 \quad \text{almost surely}.
\end{equation}
This implies that there exists a subsequence of iterates $\{\theta_{t_k}\}$ such that $\|\operatorname{grad}_{\mathbb{S}^{d-1}} J(\theta_{t_k})\| \to 0$.
\end{theorem}

\begin{proof}
Since $J(\mathbf{w})$ is $L$-smooth on the bounded domain (A1, A2), we invoke the standard non-convex SGD convergence analysis. From the $L$-smoothness inequality:
\begin{equation}
    J(\mathbf{w}_{t+1}) \le J(\mathbf{w}_t) + \langle \nabla J(\mathbf{w}_t), -\eta_t g_t \rangle + \frac{L}{2} \eta_t^2 \|g_t\|^2
\end{equation}
Taking the expectation conditioned on $\mathcal{F}_t$ and using A3:
\begin{equation}
    \mathbb{E}[J(\mathbf{w}_{t+1}) | \mathcal{F}_t] \le J(\mathbf{w}_t) - \eta_t \|\nabla_{\mathbf{w}} J(\mathbf{w}_t)\|^2 + \frac{L}{2} \eta_t^2 (\sigma^2 + \|\nabla_{\mathbf{w}} J(\mathbf{w}_t)\|^2)
\end{equation}
Rearranging the terms:
\begin{equation}
    \mathbb{E}[J(\mathbf{w}_{t+1}) | \mathcal{F}_t] \le J(\mathbf{w}_t) - \eta_t \left(1 - \frac{L}{2}\eta_t\right) \|\nabla_{\mathbf{w}} J(\mathbf{w}_t)\|^2 + \frac{L\sigma^2}{2} \eta_t^2
\end{equation}
Since $\eta_t \to 0$, there exists $t_0$ such that for all $t \ge t_0$, $1 - \frac{L}{2}\eta_t \ge \frac{1}{2}$. Summing over $t$, taking total expectation, and utilizing the lower bound (A4) and step size conditions ($\sum \eta_t^2 < \infty$), we apply the Robbins-Siegmund theorem (or standard stochastic approximation arguments) to conclude that the squared gradient norm is summable:
\begin{equation}
    \sum_{t=t_0}^{\infty} \frac{\eta_t}{2} \mathbb{E}[\|\nabla_{\mathbf{w}} J(\mathbf{w}_t)\|^2] < \infty
\end{equation}
Since $\sum \eta_t = \infty$, the gradient norm sequence cannot stay bounded away from zero. Thus:
\begin{equation}
    \liminf_{t \to \infty} \|\nabla_{\mathbf{w}} J(\mathbf{w}_t)\| = 0 \quad \text{a.s.}
\end{equation}
By definition of limit inferior, there exists a subsequence $\{t_k\}$ such that $\lim_{k \to \infty} \|\nabla_{\mathbf{w}} J(\mathbf{w}_{t_k})\| = 0$.
Using the norm relationship (Eq. \ref{eq:norm_relation}) and the Bounded Norm assumption (A1), we have:
\begin{equation}
    \|\operatorname{grad}_{\mathbb{S}^{d-1}} J(\theta_{t_k})\| = \|\mathbf{w}_{t_k}\|_2 \cdot \|\nabla_{\mathbf{w}} J(\mathbf{w}_{t_k})\| \le M \cdot \|\nabla_{\mathbf{w}} J(\mathbf{w}_{t_k})\|
\end{equation}
Taking the limit as $k \to \infty$:
\begin{equation}
    \lim_{k \to \infty} \|\operatorname{grad}_{\mathbb{S}^{d-1}} J(\theta_{t_k})\| \le M \cdot 0 = 0
\end{equation}
Thus, a subsequence of the Riemannian gradients converges to zero, which proves $\liminf_{t \to \infty} \|\operatorname{grad}_{\mathbb{S}^{d-1}} J(\theta_t)\| = 0$.
\end{proof}

\section{Baseline Adaptation for Session Context}
\label{app:baseline}
 To strictly align the baselines with our session-based search scenario (utilizing both query $q$ and user history $H$), we applied the following adaptation strategies:

Sparse Retrieval (BM25, DocT5Query): We employed a context-aware query expansion strategy. Specifically, we constructed an enriched query string by concatenating the current query text with the titles of the user's historically clicked items.

Dense Retrieval (DSSM-T5, Sentence-T5): We modified the input to the query encoder. Instead of feeding only the current query, we concatenated the query and history titles into a composite text sequence, enabling the T5 encoder to generate a context-aware dense representation.

Generative Retrieval (UniSearch, TIGER): For UniSearch, we utilized a prompt engineering approach where the query and history are concatenated into a unified text prompt. For TIGER, we mapped the current query text into its corresponding Semantic ID (SID) sequence and appended it to the history item SID sequence, creating an extended ID context for autoregressive prediction.

DSI (Adapted with Reranking): Standard DSI lacks a mechanism to directly encode long interaction histories during generation. To prevent underestimating its performance, we enhanced DSI with a two-stage "Generate-then-Rerank" framework. The model first generates candidate document cluster IDs based on the query (Recall Stage). Then, a dense reranker (utilizing the T5 encoder) takes the full context (Query + History) to re-score the candidates (Reranking Stage).

\section{Inference Efficiency Analysis}
\label{app:efficiency}

We conduct a comprehensive evaluation of the practical serving efficiency of DGI compared to standard Dense Retrieval (DR) pipelines. All benchmarks are performed on the AE-PV dataset using a single NVIDIA A100 (80GB) GPU.

DGI achieves an average end-to-end latency of \textbf{67.34 ms} per query with a beam size of 10. The process comprises three sequential stages:
\begin{enumerate}
    \item \textbf{Context Encoding ($\sim$12 ms):} The T5-based encoder processes the concatenated query-history sequence. This step is parallelizable and benefits significantly from batching.
    \item \textbf{Generative Hypersphere Search ($\sim$52 ms):} This is the computational bottleneck. As an auto-regressive process, the latency scales linearly with the code length $L$ and the beam size $B$. However, unlike DR, this cost is corpus- independent. It does not grow as the item corpus size $|\mathcal{I}|$ increases.
    \item \textbf{Dense Geometric Verification ($\sim$3 ms):} Calculating the Scaled Cosine Similarity between the query vector and the decoded item embeddings is computationally negligible, involving only highly optimized matrix multiplications.
\end{enumerate}

\section{Limitations and Future Work}
While DGI demonstrates superior performance and theoretical soundness, we acknowledge certain limitations that open avenues for future research:

\textbf{Dependence on Content Quality:}
DGI relies on the initial item encoder to initialize the semantic space. If the raw item features contain severe noise or ambiguity, the ``Semantic Alignment'' enforced by our Weight Sharing mechanism might propagate this noise into the retrieval logic. Investigating robust pre-training objectives that are less sensitive to input noise is a promising direction.

\textbf{Static Codebook Capacity:}
Currently, the codebook size $K$ and depth $m$ are fixed hyperparameters. In streaming scenarios with rapidly growing item corpora, a fixed capacity might eventually saturate. Developing \textit{Dynamic DGI} that can adaptively expand the codebook or tree depth in an end-to-end manner remains an open challenge.

\end{document}